\documentclass[12pt]{article}
\usepackage[utf8]{inputenc}
\usepackage{amsmath}

\usepackage{color}
\usepackage{amsmath, amssymb, amsthm}
\usepackage{mathrsfs}
\usepackage{mathtools}
\usepackage[noabbrev,capitalize,nameinlink]{cleveref}
\crefname{equation}{}{}
\usepackage{fullpage}
\usepackage[backend=biber,
style=numeric,
citestyle=authoryear, doi=false,isbn=false,url=false,eprint=false]{biblatex}
\usepackage{graphics}
\usepackage{pifont}
\usepackage{tikz}
\usepackage{bbm}
\usepackage[T1]{fontenc}

\usetikzlibrary{arrows.meta}

\usepackage{environ}
\usepackage{framed}
\usepackage{url}
\usepackage[linesnumbered,ruled,vlined]{algorithm2e}
\usepackage[noend]{algpseudocode}
\usepackage[labelfont=bf]{caption}
\usepackage{framed}
\usepackage[framemethod=tikz]{mdframed}
\usepackage{appendix}
\usepackage{graphicx}
\usepackage[textsize=tiny]{todonotes}
\usepackage[most]{tcolorbox}
\usepackage{enumerate}
\usepackage{verbatim}
\allowdisplaybreaks[1]

\apptocmd{\sloppy}{\hbadness 10000\relax}{}{} 

\crefname{algocf}{Algorithm}{Algorithms}

\crefname{equation}{}{} 
\crefname{conjecture}{Conjecture}{Conjectures} 
\AtBeginEnvironment{appendices}{\crefalias{section}{appendix}} 

\usepackage[color,final]{showkeys} 

\colorlet{refkey}{orange!20}
\colorlet{labelkey}{blue!30}

\crefname{algocf}{Algorithm}{Algorithms}

\numberwithin{equation}{section}
\newtheorem{theorem}{Theorem}[section]

\newtheorem{lemma}[theorem]{Lemma}

\crefname{claim}{Claim}{Claims}

\newtheorem*{question*}{Question}

\theoremstyle{definition}
\newtheorem{definition}[theorem]{Definition}

\newtheorem*{definition*}{Definition}

\theoremstyle{remark}

\newcommand{\mb}{\mathbb}

\newcommand{\NN}{\mb{N}}

\allowdisplaybreaks

\title{Social Surplus Maximization in Sponsored Search Auctions Requires Communication \footnote{I thank Yannai Gonczarowski for helpful comments.}}
\author{Suat Evren \footnote{Evren: Departments of Mathematics, Computer Science, and Economics, Massachussets Institute of Technology; email: evrenis@mit.edu.} \\ MIT}
\date{May 2023}

\begin{document}

\maketitle

\begin{abstract}
    We show that computing the optimal social surplus requires $\Omega(mn)$ bits of communication between the website and the bidders in a sponsored search auction with $n$ slots on the website and with tick size of $2^{-m}$ in the discrete model, even when bidders are allowed to freely communicate with each other. 
\end{abstract}

\textbf{Keywords:} communication complexity, sponsored search auctions, online advertising, rearrangement inequality, optimal social surplus

\textbf{2012 ACM Classification:} Theory of computation $\rightarrow$ Communication complexity; Theory of computation $\rightarrow$ Computational pricing and auctions




\section{Introduction}

In the classical model of sponsored search auctions as studied in by \citeauthor{EOS07} (\citeyear{EOS07}), there is a website that has a certain number of slots allocated for advertisement. Each advertisement slot has a publicly known click-through rate between 0 and 1 that is equal to the proportion of the users who will click on the link, and the website allocates the slots by conducting an auction. Bidders are typically companies or publishers who want to advertise their products or content and are assumed to not care about impression but only the number of clicks they get. Also, they are assumed to be risk-neutral, i.e. each have a utility function that is linear in the number of clicks they get. In general, there are two primary questions of economic interest that are well-studied in the field of mechanism design: (1) How should the website allocate the slots so that the social surplus will be maximized? (2) How much the website should charge the bidders so that the mechanism will be strategy-proof? Answering the first question in this model is a direct application of the Rearrangement Inequality (\cite{HLP34}). For the second question of interest, Vickrey-Clarke-Groves (VCG) mechanism (\cite{Vickrey61, Clarke71, Groves73}) fully characterizes the answer. 

We answer the following related question in this paper: Assuming the bidders bid truthfully in this mechanism,\footnote{Note that requiring strategy-proofness restricts uniquely determines the mechanism as VCG mechanism. Here, we require truthfulness for the website to be able decide on what is social surplus maximizer given the bids. In fact, this assumption can be relaxed to requiring monotonicity. That is, an advertiser with a higher valuation per-click bids higher. This does not change our analysis.} how much communication between the bidders and the website is necessary for the website to compute the maximum social surplus, that is the social surplus obtained from the optimal allocation, assuming the bidders can freely communicate with each other? \footnote{We could ask two different versions of this question, each of which are stronger. First, we could ask the question that how much communication between the bidders and the website is necessary for the website to determine the optimal allocation rather than maximum social surplus. Second, we could remove the assumption that bidders can freely communicate with each other. The former is an easier question compared to the one we are taking on, and can be easily answered by using classical methods. The latter requires a multi-party communication model, which is out of the scope of this paper. Yet, the lower bound we give in this paper is tight, which also gives a lower bound for the latter question.} To measure complexity, we use Yao's (\citeyear{Yao79}) model of two-party communication complexity. In particular, we show the following:
\begin{theorem}[Informal, see \ref{formal main theorem} for the formal version] \label{informal main theorem}In a sponsored search auction, for the bidders to send their bids directly to the website provides almost the least amount of communication even if the bidders are allowed to freely communicate between each other and even if the task is to only compute the maximum social surplus rather than finding the exact allocation of slots. 
\end{theorem}

We say "\emph{almost} the best" since the lower bound we provide is asymptotically equal to the communication used in the canonical mechanism in which the bidders send their bids directly to the website. This result is important for three primary reasons. First, communication complexity provides a lower bound on the computational complexity. Although in this paper, we do not focus on computational complexity. Second, the result holds in a permitting setting where we allow bidders to freely communicate with each other for the sake of decreasing communication with the website, moreover with an objective as simple as computing the optimal social surplus rather than finding the optimal matching. Third, it shows that the canonical protocol that is used in practice provides the least amount of communication between the parties. We give the formal version of the main theorem in Section 3 after giving the details of our model.

\subsection{Related Literature}

Sponsored search auctions with a focus on generalized second-price auctions was studied by \citeauthor{EOS07} (\citeyear{EOS07}). Note that our analysis in this paper is not mechanism specific but assumes that bidders are truthful. If we require strategy-proofness, this restricts our space to the VCG mechanism (\cite{Vickrey61, Clarke71, Groves73}). Also note that the matching model we use can be mapped into the basic matching model with utilities-per-matching described by \citeauthor{SS71} (\citeyear{SS71}), and most results can be generalized to that model. 

The communication requirements of indivisible good allocation mechanisms have been extensively analyzed in the previous literature. \citeauthor{Segal07} (\citeyear{Segal07}) studied the communication complexity of social choice rules in a general setup, whose applications include finding the surplus maximizing allocation in combinatorial auctions and also includes finding the stable matching in many-to-one matchings. Focusing on one-to-one matchings, \citeauthor{GNOR19} (\citeyear{GNOR19}) proved that the same lower bound observed in \citeauthor{Segal07} (\citeyear{Segal07}) still holds even when the agents in the same side of the market is allowed to communicate with each other at no expense. In this paper, we focus our attention on the sponsored search auctions. We prove a lower bound in the computation of the surplus maximization rather than finding the optimal allocation, and we also allow the bidders to communicate with each other at no expense. We also show that our lower bound is asymptotically tight, which is not the case for the previously mentioned papers. We would like to note that our proofs are much more concise compared to the proofs of the results from the aforementioned papers.

The problem of communicating information was first discussed by \citeauthor{Hayek45} (\citeyear{Hayek45}) in the economic literature. Some seminal papers that developed a framework to address this problem include, but not limited to, \citeauthor{MR74} (\citeyear{MR74}) and \citeauthor{Hurwicz77} (\citeyear{Hurwicz77}). Note that the model considered in these papers are different from the model we use. Namely, they use \emph{message space dimension} as a proxy for the communication complexity, and the functions they work with are continuous.

 The rest of the paper is organized as follows: In Section \ref{model}, we first give the preliminaries from communication complexity, then we describe the classical model of sponsored search auctions. In the same section, we give the exact details of our modeling and the problem definition. In Section \ref{main results}, we give the formal versions of our main results and their proofs. In Section \ref{conclusion}, we conclude.

\section{The Model and Preliminaries}\label{model}

\subsection{Communication Complexity}

We use the two-party communication complexity model introduced by \citeauthor{Yao79} (\citeyear{Yao79}), whose description is as follows: Consider a situation where two agents, Alice and Bob, want to answer a question by combining the information they separately hold with the least amount of communication. In particular, given $N \in \NN$ and a Boolean function $f : \{0,1\}^N \times \{0,1\}^M \to \{0,1\}$, Alice has an input $x \in \{0,1\}^N,$ and Bob has an input $y \in \{-1,1\}^M$, and they want to compute $f(x,y)$. Computing a function in this context means that in the end, at least one of the agents will know $f(x,y)$, and then by using a single bit, they can send the result to the other agent as well. A \textbf{communication protocol} $P$ is defined to be a distributed algorithm of two processors to compute $f$ for any given pair of inputs $(x,y)$. \textbf{Communication complexity of a protocol} $P$ is defined to be the number of bits need to be exchanged between two parties in protocol $P$ in the worst case for any given $(x,y)$, and it is denoted by $C(P)$. \textbf{Communication complexity of a function} is equal to the minimum communication complexity of any protocol that computes $f$, and it is denoted by $C(f)$. 

Communication complexity is also generalized to the randomized setting, but in this paper we will only consider the deterministic case to keep the proofs concise and insightful. Also, note that some of Yao's definitions are slightly different than ours; we use the modern definitions.

In our proof of the lower bound, we will use two main techniques: (1) reduction, (2) the fooling set method. A fooling set is defined as the following:

\begin{definition}
A set of input pairs $\{(x_1, y_1), (x_2,y_2), \cdots, (x_k,y_k)\}$ is called a \emph{fooling set} with respect to $f$ if there exists $b \in \{0,1\}$ such that
\begin{enumerate}
    \item $\forall i \in [k], f(x_i,y_i) = b$, and 
    \item $\forall i \neq j$, either $f(x_i,y_j) \neq b$ or $f(x_j,y_i) \neq b$.
\end{enumerate}
\end{definition}

In particular, will make use of the following well-known result:
\begin{lemma} \label{fooling set theorem}
If there exists a fooling set of size $k$ with respect to $f$, then
\begin{equation*}
    C(f) \geq \log_2 k.
\end{equation*}
\end{lemma}

For a survey on communication complexity with modern terminology, as well as the proofs and more detailed discussions of the fundamental techniques in the field including a proof of Lemma \ref{fooling set theorem}, see \citeauthor{Kushilevitz97} (\citeyear{Kushilevitz97}). For a more detailed survey of the field, see \citeauthor{KN06} (\citeyear{KN06}).

\subsection{Sponsored Search Auctions and Our Model}

In the classic model of sponsored search auctions, there are $n$ slots $S_1, S_2, \cdots, S_n$ on a web page, and $n$ advertisers $A_1, A_2, \cdots, A_n$. Each slot has a fixed click-through rate $ \alpha_i \in [0,1]$, \footnote{Click-through rates in practice can be determined by using statistical learning on historical data.} which is the (expected) proportion of the users who will click on the link. It is assumed in the general model that advertisers do not care about impression; each advertiser $i$ has a valuation $v_i$ per click, and her utility function is $u_i(S_j) = \alpha_j \cdot v_i $ if she is assigned to slot $S_i$. Without loss of generality, we can assume that $\alpha_1 \geq 
\alpha_2 \geq \cdots \geq \alpha_n$ and $v_1 \geq 
v_2 , \cdots \geq v_n$. Let $\sigma: [n] \to [n]$ be a one-to-one function such that for all $i \in [n]$, $S_{\sigma(i)}$ is the slot $A_i$ is assigned to. Then, the social surplus is equal to 
\begin{equation}
    \sum_i v_i \cdot \alpha_{\sigma(i)} 
\end{equation}
where $x_i$ is the click-through rate of the slot bidder-$i$ is assigned to.

A natural question of economic interest is to ask which allocation maximizes the social surplus, whose answer is also intuitive. Between two advertisers, we need to assign the advertiser with the greater valuation to the slot with the greater click-through rate. The following result is a direct application of the Rearrangement Inequality (\cite{HLP34}), whose proof is by a swapping trick.

\begin{lemma} \label{rearrangement ineq}
Social surplus is maximized when $\sigma(i) = i$ for every $i \in [n]$.
\end{lemma}

As in most other economic models, it is customary to assume that valuations, bids, and click-through rates are continuous in the classical model of sponsored search auctions. However, in our model, we assume that there is a \textbf{tick size} for values and click-through rates. In particular, we assume that values of bidders per click is a non-negative integer smaller than $2^m$ for a given positive integer $m$, and that each click-through-rate is a non-negative integer multiple of $2^{-m}$. As we have $n$ slots, $2^m$ should be thought as a number much larger than $n$. For example, $m$ would be $64$ for a 64-bit signed computer. The main reason we need this assumption is that in our problem definition to compute communication complexity, we need to represent valuations and click-through-rates as discrete values. In fact, this assumption makes our model more representative of the reality since the bidding process is performed by using computers in practice.

Second, we assume in our model that the bidders bid truthfully. Only then it makes sense to talk about maximizing the surplus. In particular, so long as they bid truthfully, we do not care what mechanism is used to charge the bidders and how much they are charged in the end. Although we know that the VCG mechanism is the unique strategy-proof mechanism, this assumption can in fact be relaxed to requiring monotonicity. That is, an advertiser with a higher valuation per-click bids higher. Since it does not change our analysis, we assume truthfulness for the sake of simplicity in the notation.

We are going to make one further simplification in our problem definition. Instead of considering click-through-rates $1 > \alpha_1 \geq \cdots \geq \alpha_n \geq 0$ with tick size of $2^{-m}$, we will consider non-negative integers $2^m > a_1 \geq a_2 \geq \cdots \geq a_n \geq 0$ with tick size of $1$. Note that there is a one-to-one correspondence between the sequences of click-through-rates $\{\alpha_i\}_{i \in [n]}$ and sequences $\{a_i\}_{i \in [n]}$ with the mapping defined as $a_i = 2^m \alpha_i$ for all $i \in [n]$ for any given sequence $\{\alpha_i\}_{i \in [n]}$. Also, maximizing $\sum_i v_i \cdot \alpha_{\sigma(i)}$ is equivalent to maximizing $\sum_i v_i \cdot a_{\sigma(i)}$. Therefore, this simplification does not make the problem either more or less general, but it only makes the computation easier.

Now, we can give the formal problem definition. For a given pair of positive integers $(m,n)$, we define the \textbf{Optimal Social Surplus Inequality Problem (OSSI)} as follows: There are two parties, the website ($W$) and the bidders ($B$). $W$ is given the sequence $\{a_i\}_{i \in [n]}$ defined as above, and a constant $n \cdot 2^{2m} > c \geq 0$. \footnote{The reason $c < n2^{2m}$ is that $\sum_i v_i \cdot a_{\sigma(i)} < n2^{2m}$ because each term in the sum is less than $2^{2m}$.} $B$ is given the integer sequence $2^m > b_1 \geq b_2 \geq \cdots \geq b_n \geq 0$ corresponding to the bids. Note that these inputs can be represented by $mn + 2m\log_2 n$ and $mn$ bits, respectively. The Boolean function $f : \{0,1\}^{mn + 2m\log_2 n} \times \{0,1\}^{mn} \to \{0,1\}$ is defined as $f(x,y) = 1$ if the maximum social surplus is greater than or equal to $c$, and $f(x,y) = 0$ otherwise. 

In particular, under the assumption that $a_i$ and $b_i$ are decreasing sequences, \footnote{Note that restricting the input space to decreasing sequences does not make any difference because after receiving the input, both $W$ and $B$ can rearrange their inputs to obtain a decreasing sequence without any communication.} the maximum value $\sum_i b_i \cdot a_{\sigma(i)}$ can attain is equal to $\sum_i b_i \cdot a_{i}$ by Lemma \ref{rearrangement ineq}. Therefore, $f(x,y) = 1$ if and only if $\sum_i b_i \cdot a_{i} \geq c$.

Finally, we note that in our model, we take the number of slots to be equal to the number of advertisers. In practice, the number of advertisers is possibly much larger. However, the simplified model we use is equivalent to the general model. In particular, without knowing the exact click-through rates of the slots but only knowing the number $n$ of slots on the website, the advertiser side of the market can find the top $n$ advertisers whose valuations of a click is the greatest, and eliminate all of the remaining advertisers in the first place without requiring any communication between two sides of the market.

\section{Main Results and Their Proofs}\label{main results}
With our model in hand, we can now state the formal version of our main theorem.

\begin{theorem}[Formal Version]\label{formal main theorem}
 Let the tick size in the discrete model be $2^{-m}$, and let there be $n$ slots on the website to be allocated for ads. Then,
 \begin{equation*}
   \log_2 \binom{2^m + n -1}{n} \geq C(OSSI) \geq \log_2 \binom{2^m + n -1}{n} - 2m \log_2 n.   
 \end{equation*}

\end{theorem}

Under the assumption that $2^m \gg n$, upper and lower bounds in our result are quite close and both are in the order of $\Omega(mn)$. Therefore, our main theorem implies that $C(OSSI) = \Omega(mn)$. Noting that when all the bidders send their bids to the website, total number of bits communicated is equal to $mn$. Therefore, it implies that even if the bidders are allowed to freely communicate with each other, almost the best they could do in order to find the optimal allocation is to send their bids directly to the website.

The proof of the lower bound in Theorem \ref{formal main theorem} constitutes two parts. First, we introduce another problem which we call the \textbf{Permutation Problem (PP)}, and find a lower bound on its communication complexity by using the fooling set trick. In the second part, we construct a reduction from PP to OSSI. In particular, by using the optimal protocol for OSSI as a building block, we construct a protocol for the Permutation Problem, which gives us an upper bound on the communication complexity of PP in terms of the communication complexity of the OSSI. Therefore, together with the first step, we obtain a lower bound on the communication complexity of OSSI. 

In fact, we do not only find a lower bound on the communication complexity of PP, but fully characterize it. That is, we show that the lower bound is achievable by a certain protocol. And then, by using the same exact protocol for OSSI, we obtain the upper bound in our main theorem. In the next subsection, we define the permutation problem and compute its communication complexity. In the subsection following it, we give a reduction from PP to OSSI, and we end it by completing the proof of Theorem \ref{formal main theorem}.

\subsection{The Permutation Problem (PP)}

For given positive integers $m,n$, we define the \textbf{Permutation Problem (PP)} as follows: Alice is given a sequence of non-negative integers $a_1, a_2, \cdots , a_n$ where $0 \leq a_i < 2^m$ for all $i$. Similarly, Bob is also given a sequence of non-negative integers $b_1, b_2, \cdots , b_n$ where $0 \leq b_i < 2^m$ for all $i$. They try to decide whether $(a_1, a_2, \cdots, a_n)$ is a permutation of $(b_1, b_2, \cdots, b_n)$. In particular, they try to compute the function $f_{PP} : \{0,1\}^{mn} \times \{0,1\}^{mn} \to \{0,1\}$, defined as $f_{PP}(x,y) = 1$ if and only if the sequence corresponding $x$ is a permutation of the sequence corresponding $y$.

We use the following lemma to prove a lower bound on the communication complexity of PP. Let us denote by $G$ the set of $\gamma \in \{0,1\}^{mn}$ such that the corresponding sequence of $\gamma$ is a non-decreasing sequence.

\begin{lemma}\label{a fooling set for PP}
$\{(\gamma, \gamma) : \gamma  \in G \}$ is a fooling set of size $2^m + n -1 \choose n$. 
\end{lemma}
\begin{proof}
The proof that it is a fooling set is immediate from the fact that two non-decreasing sequences are permutations of each other only if these are the two sequences are exactly the same. The size of the set is the number of solutions to the equation $z_0 + z_1 + z_2 + \cdots + z_{2^m-1} = n$ where $z_i$ is equal to the number of elements in the corresponding sequence that is equal to $i$. This is a classical stars and bars problem, which has $2^m + n -1 \choose n$ different solutions in the set of non-negative integers.
\end{proof}

The following theorem fully characterizes the Permutation Problem.

\begin{theorem} \label{communication complexity of PP}
$C(PP) = \log_2 {2^m + n -1 \choose n}.$
\end{theorem}

\begin{proof}
We prove the theorem by showing that $C(PP) \geq \log_2 {2^m + n -1 \choose n}$ and $C(PP) \leq \log_2 {2^m + n -1 \choose n}.$ The lower bound immediately follows from Lemma \ref{fooling set theorem} and Lemma \ref{a fooling set for PP}.

For the upper bound, we show that a canonical protocol for PP has communication complexity of $\log_2 {2^m + n -1 \choose n}$. As in the proof of Lemma \ref{a fooling set for PP}, the number of sequences that none of which are permutations of each other is equal to the number of solutions to the equation 
$$z_0 + z_1 + z_2 + \cdots + z_{2^m-1} = n$$
in the set of non-negative integers, which is $\log_2 {2^m + n -1 \choose n}$. We define the optimal protocol for PP in a canonical way as follows. Fix a bijective function $g: G \to \left[\log_2 {2^m + n -1 \choose n}\right]$ that is known to both Website and Bidders before the communication stars as a part of the protocol. Then, for any input $y$, Bidders sends $g(y)$ to Website by using $\log_2 {2^m + n -1 \choose n}$. Finally, Website recovers $y$ from $g(y)$ by computing $g^{-1}(g(y))$, and computes $f_{PP}(x,y)$.
\end{proof}

\subsection{Reduction from the Permutation Problem}
\paragraph*{Reduction from PP to OSSI.} We define the reduction as follows. For any input $(x,y)$ for the permutation problem, Bob computes $b_1^2 + b_2^2 + \cdots + b_n^2$ and sends the result to Alice by using $2m \log_2n$ bits, and inputs $(b_1, b_2, \cdots, b_n)$ to the protocol for OSSI as a proxy of the Bidders. Alice finds her corresponding sequence $a_1, a_2, \cdots, a_n$, computes $ \max \{\sum_{i = 1}^{n} a_i^2, \sum_{i = 1}^{n} b_i^2\}$, and inputs $(a_1, a_2, \cdots, a_n, \max \{\sum_{i = 1}^{n} a_i^2, \sum_{i = 1}^{n} b_i^2\})$ to the protocol for OSSI as a proxy of the website. Then, our protocol for PP outputs the result we get from this query to the OSSI protocol. 

Now, let us prove why this protocol outputs 1 if and only if two corresponding sequences are permutations of each other. 

\begin{lemma}\label{reduction-1}
If the two sequences are permutations of each other, then the protocol outputs 1.
\end{lemma}
\begin{proof}
Note that if the two sequences are permutations of each other, then by Lemma \ref{rearrangement ineq} the maximum value that a rearranged sum can attain is equal to 
\begin{equation}
  \sum_{i = 1}^{n} a_i b_i = \sum_{i = 1}^{n} a_i^2 = \sum_{i = 1}^{n} b_i^2.  
\end{equation} 
Therefore, it will be the equality case and the query to the OSSI protocol will output 1.
\end{proof}

\begin{lemma}\label{reduction-2}
If the protocol outputs 1, then the two sequences are permutations of each other.
\end{lemma}
\begin{proof}
Let $(a_i^*)_{1 \leq i \leq n}$ be the rearranged version of $(a_i)_{1 \leq i \leq n}$ such that $a_1^* \geq a_2^* \geq \cdots \geq a_n^*$. Let $(b_i^*)_{1 \leq i \leq n}$ be defined similarly. By the Rearrangement Inequality, the maximum value a rearranged sum can attain is $\sum_{i =1}^{n} a_i^* b_i^*$. 

Since the query to the OSSI protocol outputs 1,
\begin{equation}
    \max \{\sum_{i = 1}^{n} a_i^2, \sum_{i = 1}^{n} b_i^2\} \leq \sum_{i =1 }^{n} a_i^* b_i^*
\end{equation}
which implies that
\begin{equation}\label{query ineq}
    \sum_{i =1}^{n} a_i^2 \sum_{i =1}^{n} b_i^2 
    \leq 
    (\max \{\sum_{i = 1}^{n} a_i^2, \sum_{i = 1}^{n} b_i^2\})^2
    \leq 
    (\sum_{i =1}^{n} a_i^* b_i^*)^2
\end{equation}
On the other hand, by Cauchy-Schwarz,
\begin{equation} \label{Cauchy ineq}
      (\sum_{i =1}^{n} a_i^* b_i^*)^2 \leq \sum_{i =1}^{n} (a_i^*)^2 \sum_{i =1}^{n} (b_i^*)^2 = \sum_{i=1}^{n} a_i^2 \sum_{i=1}^{n} b_i^2 
\end{equation}
Therefore, the equality case must be satisfied in all of the inequalities in \ref{query ineq} and \ref{Cauchy ineq}. Equality case in the Cauchy-Schwarz inequality is satisfied if and only if the vector $(a_1^*, a_2^*, \cdots, a_n^*)$ is a scalar multiple of the vector $(b_1^*, b_2^*, \cdots, b_n^*)$. For the equality case in \ref{query ineq},
\begin{equation}
    \sum_{i=1}^{n} a_i^2 = \sum_{i=1}^{n} b_i^2 
\end{equation}
which implies that $c=1$, and thus $a_i^* = b_i^*$ for all $i \in [n]$, which implies that $(a_i)_{1 \leq i \leq n}$ and $(b_i)_{1 \leq i \leq n}$ are permutations of each other.
\end{proof}

Lemma \ref{reduction-1} and Lemma \ref{reduction-2} together show that the reduction we defined from PP to OSSI is indeed a valid reduction. Now, we are ready to complete the proof of our main result.

\begin{proof}[Proof of Theorem \ref{formal main theorem}]
The communication complexity of the protocol for PP given above by using the optimal protocol for OSSI as black-box is equal to $C(OSSI) + 2m \log_2 n$ because before querying the OSSI protocol, the Bob sends Alice the sum $\sum_{i = 1}^{n} b_i^2$ by using $2m \log_2 n$ bits. By Theorem \ref{communication complexity of PP}, it $C(OSSI) \geq \log_2 {2^m + n -1 \choose n} - 2m \log_2 n$.

For the upper bound, we need to give a protocol. It is easy to see that the same exact protocol we used for the upper bound in the proof of Theorem \ref{communication complexity of PP} also works for OSSI, which completes the proof.
\end{proof}

\section{Discussion}\label{conclusion}

In this paper, we seek to answer the question of how much communication is required between advertisers and the website for finding the optimal social surplus in the sponsored search auctions, otherwise known as online ads market, which was studied by \citeauthor{EOS07} (\citeyear{EOS07}) in the literature. In this context, we analyze the discrete communication complexity of two related problems by using the model described in \citeauthor{Yao79} (\citeyear{Yao79}), namely the Permutation Problem (PP) and the Optimal Social Surplus Inequality (OSSI). We first find the exact communication complexity of PP by showing that when two agents are given a tuple of numbers and want to find whether the tuples they have are permutations of each other with the least amount of communication, the best they can do is for one of them to send the whole tuple to the other. Second, by using a reduction from OSSI to PP, we prove a tight lower bound on the complexity of the communication required between the website and the bidders in the problem of finding the maximum social surplus in a sponsored search auction. 

Our results can similarly be extended to a setting where the advertisers observe decreasing marginal returns, i.e. their utility function is concave in the number of clicks they receive. Further extensions of this paper include the analysis in the continuous setting, indeterministic and probabilistic settings. The results may also be generalized to the matching model described by \citeauthor{SS71} (\citeyear{SS71}).

An interesting question to answer would be to measure the communication complexity of the revenue under VCG mechanism (\cite{Vickrey61, Clarke71, Groves73}). We believe that similar proof techniques would extend, but perhaps with slightly more tedious math. 

It would also be interesting to analyze it under a Bayesian setting where the website ex-ante does not know the bids but it has a belief of distributions where it tries to estimate the ex-ante expected revenue. This question could have more economic implications since the website needs to have an estimate of the revenue before it commits to enter the market.

\printbibliography

\end{document}